\documentclass[twoside,leqno]{article}

\usepackage{fullpage}
\usepackage{bbm}
\usepackage{url}
\usepackage{blkarray}
\usepackage{complexity}
\usepackage{booktabs}
\usepackage{amsmath,amssymb}
\usepackage{amsthm}

\usepackage{nicefrac}
\usepackage{microtype}
\usepackage{calc}
\usepackage{enumerate}
\usepackage[usenames,dvipsnames]{xcolor}
\usepackage[colorlinks,citecolor=blue,linkcolor=BrickRed]{hyperref}
\usepackage{algorithm}
\usepackage[labelfont=bf]{subcaption}
\usepackage[font=scriptsize,labelfont=bf]{caption}
\usepackage{appendix}
\usepackage[noend]{algpseudocode}
\usepackage{xr}
\usepackage{array}
\usepackage{xspace}
\usepackage{hyperref}
\usepackage[capitalise,nameinlink]{cleveref}
\usepackage{parskip}
\usepackage{chngcntr}
\usepackage{mathtools}
\usepackage{tikz}
\usepackage{relsize}
\usepackage{nameref}
\usepackage{thmtools}
\usepackage{thm-restate}

\makeatletter
 \allowdisplaybreaks
\makeatother

\DeclareMathOperator*{\argmin}{argmin}

\DeclareMathOperator*{\expectation}{\mathbb{E}}
\let\poly\relax
\DeclareMathOperator*{\poly}{poly}
\DeclareMathOperator*{\probability}{Pr}

\renewcommand\R{\mathbb{R}}

\newcommand\expect[1]{\expectation\left[ #1 \right]}
\newcommand\expectarg[2]{\expectation_{#2}\left[ #1 \right]}
\newcommand\prob[1]{\probability\left[ #1 \right]}

\newcommand\mutcov{\mathcal{I}_f}

\newcommand\opt{\textsc{Opt}\xspace}

\counterwithin{equation}{section}

\newcommand\sscText{$\ComplexityFont{SubCov}$\xspace}
\newcommand\osscText{$\ComplexityFont{OSubCov}$\xspace}

\newcommand\ssc{\hyperref[sec:introduction]{\sscText}\xspace}
\newcommand\ossc{\hyperref[prob:ossc]{\osscText}\xspace}

\usepackage{eqparbox}

\newcommand\fmin{f_{\min}}
\newcommand\fmax{f_{\max}}

\newcommand\cmin{c_{\min}}
\newcommand\cmax{c_{\max}}

\theoremstyle{plain}
\newtheorem{theorem}{Theorem}
\newtheorem{lemma}[theorem]{Lemma}

\newtheorem{fact}[theorem]{Fact}
\newtheorem{claim}[theorem]{Claim}

\theoremstyle{definition}
\newtheorem*{definition}{Definition}
\theoremstyle{definition}

\crefname{claim}{Claim}{Claims}
\crefname{theorem}{Theorem}{Theorems}

\newlength{\continueindent}
\usepackage{etoolbox}
\makeatletter
\newcommand*{\ALG@customparshape}{\parshape 2 \leftmargin \linewidth \dimexpr\ALG@tlm+\continueindent\relax \dimexpr\linewidth+\leftmargin-\ALG@tlm-\continueindent\relax}
\apptocmd{\ALG@beginblock}{\ALG@customparshape}{}{\errmessage{failed to patch}}
\makeatother

\makeatletter
\def\thm@space@setup{%
	\thm@preskip=\parskip \thm@postskip=0pt
}
\makeatother

\usepackage{etoolbox}
\usepackage{tikz}
\usetikzlibrary{tikzmark}
\usetikzlibrary{calc}

\errorcontextlines\maxdimen

\newcommand{\ALGtikzmarkcolor}{black}% customise this, if you want
\newcommand{\ALGtikzmarkextraindent}{4pt}% customise this, if you want
\newcommand{\ALGtikzmarkverticaloffsetstart}{-.5ex}% customise this, if you want
\newcommand{\ALGtikzmarkverticaloffsetend}{-.5ex}% customise this, if you want
\makeatletter
\newcounter{ALG@tikzmark@tempcnta}

\newcommand\ALG@tikzmark@start{%
	\global\let\ALG@tikzmark@last\ALG@tikzmark@starttext%
	\expandafter\edef\csname ALG@tikzmark@\theALG@nested\endcsname{\theALG@tikzmark@tempcnta}%
	\tikzmark{ALG@tikzmark@start@\csname ALG@tikzmark@\theALG@nested\endcsname}%
	\addtocounter{ALG@tikzmark@tempcnta}{1}%
}

\def\ALG@tikzmark@starttext{start}
\newcommand\ALG@tikzmark@end{%
	\ifx\ALG@tikzmark@last\ALG@tikzmark@starttext
	\else
	\tikzmark{ALG@tikzmark@end@\csname ALG@tikzmark@\theALG@nested\endcsname}%
	\tikz[overlay,remember picture] \draw[\ALGtikzmarkcolor] let \p{S}=($(pic cs:ALG@tikzmark@start@\csname ALG@tikzmark@\theALG@nested\endcsname)+(\ALGtikzmarkextraindent,\ALGtikzmarkverticaloffsetstart)$), \p{E}=($(pic cs:ALG@tikzmark@end@\csname ALG@tikzmark@\theALG@nested\endcsname)+(\ALGtikzmarkextraindent,\ALGtikzmarkverticaloffsetend)$) in (\x{S},\y{S})--(\x{S},\y{E});%
	\fi
	\gdef\ALG@tikzmark@last{end}%
}

\apptocmd{\ALG@beginblock}{\ALG@tikzmark@start}{}{\errmessage{failed to patch}}
\pretocmd{\ALG@endblock}{\ALG@tikzmark@end}{}{\errmessage{failed to patch}}
\makeatother
\usepackage[colorinlistoftodos,textsize=tiny,textwidth=2cm,color=red!25!white,obeyFinal]{todonotes}

\newcommand{\ts}{\textstyle}

\title{The Online Submodular Cover Problem}

\author{Anupam Gupta\thanks{Department of Computer Science, New York University, New York, NY 10012. Email: \texttt{anupam.g@nyu.edu}.}\ $^{\ddagger}$
	\and
	Roie Levin\thanks{Department of Computer Science, Rutgers University, Piscataway, NJ 08854. Email: \texttt{roie.levin@rutgers.edu}.}\ \thanks{The original version of this paper was written while the authors were affiliated with Carnegie Mellon University and supported by NSF awards CCF-1536002, CCF-1540541,
		and CCF-1617790.}
}
\date{}

\begin{document}

\maketitle

\begin{abstract}
	\medskip
	\small{
In the submodular cover problem, we are given a monotone submodular
function $f: 2^N \to \R_+$, and we want to pick the min-cost set $S$
such that $f(S) = f(N)$. This captures the set cover problem when $f$
is a coverage function. Motivated by problems in network monitoring and resource allocation,
we consider the submodular cover problem in an \emph{online} setting. As a concrete example, 
suppose at each time $t$, a nonnegative
monotone submodular function $g_t$ is given to us. We define
$f^{(t)} = \sum_{s \leq t} g_s$ as the sum of all functions seen so
far. We need to maintain a submodular cover of these submodular
functions $f^{(1)}, f^{(2)}, \ldots f^{(T)}$ in an online fashion;
i.e., we cannot revoke previous choices. Formally, at each time $t$ we
produce a set $S_t \subseteq N$ such that $f^{(t)}(S_t) =
f^{(t)}(N)$---i.e., this set $S_t$ is a
cover---such that $S_{t-1} \subseteq S_t$, so previously decisions to pick elements cannot be revoked.
(We actually allow more general sequences $\{f^{(t)}\}$ of submodular
functions, but this sum-of-simpler-submodular-functions case is useful for concreteness.)

We give polylogarithmic competitive algorithms for this online
submodular cover problem. The competitive ratio on an input sequence
of length $T$ is $O(\ln n \ln (T \cdot f(N) / \fmin))$, where $\fmin$ is the smallest nonzero marginal for functions
$f^{(t)}$, and $|N| = n$. For the special case of online set cover,
our competitive ratio matches that of Alon et
al.~\cite{alon2003online}, which are best possible for polynomial-time
online algorithms unless $\NP \subseteq \BPP$~\cite{korman2004use}.
Since existing offline algorithms for submodular cover are based on
greedy approaches which seem difficult to implement online, the
technical challenge is to (approximately) solve the exponential-sized
linear programming relaxation for submodular cover, and to round it,
both in the online setting. Moreover, to get our competitiveness
bounds, we define a (seemingly new) generalization of mutual information to
general submodular functions, which we call \emph{mutual coverage}; we hope this will
be useful in other contexts.}
\end{abstract}

\newpage
\setcounter{page}{1}

\section{Introduction}
\label{sec:introduction}

We consider the problem of maintaining a minimum cost submodular cover
online. 
In the offline Submodular Cover problem (\ssc) first introduced
in the seminal work of Wolsey \cite{Wolsey1982}, we are given a
monotone, submodular function $f$ over a universe $N$ along with a linear cost function $c$, and we wish to
select a minimum cost set $S \subseteq N$ such that $f(S) = f(N)$. We
may interpret $f$ as modeling an abstract notion of coverage, and $S$ as
a subset of elements guaranteeing full coverage. For example, if $N$ is
a collection of subsets over a ground set and $f(S)$ computes the
cardinality of the union of sets in $S$, \ssc exactly recovers the
classical Set Cover problem. Another example is the partial cover
problem and its capacitated version~\cite{fujito2000approximation}, in
which $f$ is the minimum of a coverage function and a threshold function. 

Like in the case of set cover, for many applications it does not suffice
to solve a single static instance of \ssc. Instead, the function
$f$---the notion of coverage---changes over time, so we need to update
our cover $S$.
For example, in the Source Location
problem~\cite{Andreev:2009:SSL:1644015.1644031}, given a network
equipped with flow demand at every vertex, we want to designate the
minimum number of vertices as servers, such that every vertex can
simultaneously route its demand from some server without violating edge
capacities. \cite{Andreev:2009:SSL:1644015.1644031} show that if $f(S)$
is the maximum demand satisfied by selecting $S$ as the set of servers,
then $f$ is submodular and the problem can cleanly be solved using the
\ssc framework. In the likely scenario that communication demand evolves
over
time, 
revoking prior decisions may be expensive or prohibitive, hence we want
that our solution be monotone---i.e., we only add servers to $S$ and not
remove any. See the related works section for further examples in which
avoiding revoking past decisions is useful in 
applications. 

We ask the natural question: is it possible to maintain a competitive
solution to an online \ssc instance as $f$ changes over time? We show
that under minimal monotonicity assumptions on the changes in $f$, the
answer is yes. Formally, the online submodular cover (\ossc) problem is
the following:

\begin{quote}
  \label{prob:ossc}
  We are given a ``time-monotone''\footnote{This ``time-monotonicity''
    means that any submodular cover for $f^{(t)}$ is also a cover for
    $f^{(t-1)}$; the functions cannot be completely unrelated to each
    other. This is a minimal requirement; we show in
    \cref{section:timeMonotonicity} that dropping it makes the problem
    intractable. On the other hand, 
    we do not explicitly require that
    $f^{(t)}(S) \geq f^{(t-1)}(S)$.} 
  sequence of monotone submodular functions $f^{(1)}, f^{(2)}, \cdots, $
  over a common universe $N$ that arrive online, as well as a fixed cost
  function $c$. Upon seeing $f^{(t)}$ at time $t$, our algorithm must
  output a set $S_t$ such that $f^{(t)}(S_t) = f^{(t)}(N)$.  Moreover,
  past decisions cannot be revoked, so we require
  $S_{t-1} \subseteq S_t$.
\end{quote}

\medskip One example of this time-monotone setting was mentioned in the
abstract: suppose each time $t$ we get a new positive monotone
submodular function $g_t$, and define $f^{(t)} := \sum_{s \leq t}
g_s$. It may be convenient for the reader to keep this example in mind
since it captures most features of the general problem.

\begin{restatable}[Main Theorem]{theorem}{mainthm}
	\label{thm:main}
  There exists an efficient randomized algorithm for the \ossc problem
  which guarantees that for each $T$, the expected cost of solution
  $S_T$ is within $O(\ln n \cdot \ln (T\cdot f(N) / \fmin))$ of the
  optimal submodular cover solution for $f^{(T)}$.
\end{restatable}
\medskip
Here 
$\fmax$ and $\fmin$
are the largest and smallest marginals for any of the functions in the
input sequence. 
When restricting to the case of online set cover,
$f_{\max} = f_{\min} =1$ and $T = f(N)$ is the number of elements of the set system, and hence this competitiveness guarantee
matches that of the best current algorithm for that
problem~\cite{alon2003online}. Moreover, if we restrict ourselves to
efficient algorithms, the hardness result of Feige and
Korman~\cite{korman2004use} says that our result is the best possible up
to constant factors unless $\NP \subseteq
\BPP$. 

Our algorithm also is interesting in the \emph{offline} setting: indeed,
if we are given a single submodular function (i.e., $T=1$), we can avoid
the $O(\ln n)$ loss that comes from solving an LP relaxation online. So
we get an $O(\ln (f(N) / \fmin))$-approximate solution via randomized
rounding a fractional LP solution. This is the first LP-rounding
algorithm known for \ssc, and it matches the (optimal) performance of
the greedy algorithm (analyzed by dual-fitting). LP rounding approaches
are powerful and versatile, since we can add side constraints to the LP
and maintain them (approximately) during rounding, so we hope that our
techniques will find applications even in the offline
setting. 

For a subclass of nonnegative, monotone, submodular functions known as $3$-increasing functions, we show slightly improved results improving the $f(N)$ parameter to $\fmax$.

\begin{restatable}{theorem}{mainthm2}
	\label{thm:main2}
	There exists an efficient randomized algorithm for the \ossc problem for $3$-increasing functions
	which guarantees that for each $T$, the expected cost of solution
	$S_T$ is within $O(\ln n \cdot \ln (T\cdot \fmax / \fmin))$ of the
	optimal submodular cover solution for $f^{(T)}$.
\end{restatable}

This is analogous to improving the approximation ratio of set cover from $\log(\text{number of elements})$ to $\log(\text{maximum set size})$.

We note that the original conference version of this paper claimed the same $O(\ln n \cdot \ln (T\cdot \fmax / \fmin))$ guarantee for the class of all monotone submodular functions, but there was a gap in the argument. See \cref{sec:version_notes} for a discussion of the changes in this version.

\subsection{Our Techniques.}
There are a few challenges to proving \hyperref[thm:main]{Theorem 1.1}. A natural
approach is to mimic the strategy of \cite{alon2003online} for the
online set cover problem. I.e., we can try to maintain a solution to an
LP relaxation of the problem, and then perform randomized rounding
online. Unfortunately, the only known LP relaxation of \ssc from
\cite{Wolsey1982} is exponential-sized and it is not known how to solve it
efficiently, nor how to analyze randomized rounding given a fractional
solution. In this paper, we show that we can overcome these concerns and
make the technique work.

We start in \cref{subsec:slowgreedy} with an $O(\ln n \ln (T \cdot f(N)
/ \fmin))$ competitive ratio algorithm for online \ssc with no efficient
runtime guarantees. This allows us to illustrate the overall approach 
of solving the \ssc covering linear program online, in conjunction
with a randomized-rounding algorithm for it.  
The rounding is the natural one, where we sample from the fractional
solution multiple times (in an online fashion); as opposed to the set
cover-style analysis which argues item-by-item, here we have to argue
about the coverage as a whole. 
In particular, we exploit the relationship between the multilinear
extension and the linearizations of submodular functions (\Cref{lem:multilinFstar})
to reinterpret the constraints in the LP as statements about
expected coverage under randomized rounding.

In \cref{subsec:betterround}, we improve the
$O(\ln n \ln (T \cdot f(N) / \fmin))$-competitiveness to an
$O(\ln n \ln (T\cdot \fmax / \fmin))$-competitiveness for $3$-increasing functions, which is much
tighter for the case where the functions $f$ are ``smoother'' and each
element has smaller marginal value compared to $f(N)$.  The intuition
behind this refined analysis is that we charge each element selected by
the algorithm to the elements of \opt that ``cover the same part of the
space''.\footnote{The analogue in the case of set cover is an
  improvement from $O(\log \text{universe-size})$ to
  $O(\log \text{max-set-size})$.  However, there is a clear notion in
  set cover of \emph{items} that are covered by the same two sets (and
  the analysis proceeds item-by-item), whereas in \ssc it is not at all
  obvious what it means for the coverage of two elements to overlap.} To
capture the overlap between two elements, we study the \textit{mutual
  coverage} $\mathcal{I}_f(A;B) := f(A) + f(B) - f(A\cup B)$, which is natural
generalization of mutual information from information theory, and
inherits properties from that literature, such as the chain-rule. It turns out that mutual coverage is the right abstract
quantity to focus on for our analyses. We also have to strengthen the submodular cover LP slightly, and use a
modified rounding-with-alterations scheme to get the tighter
result. When run offline, our rounding algorithm recovers the
$O(\ln \fmax / \fmin)$ approximation ratio of Wolsey's algorithm.

In \cref{subsec:polytime}, we make the above algorithms run in
polynomial time. The difficulty is two-fold: firstly, we do not know how
to solve the submodular cover LP even in the offline setting, since the separation
problem itself is APX-hard. Secondly, even if we could separate in the
offline setting, na{\"\i}vely solving the exponential-sized LP online
requires us to give exponentially-many constraints one at a time to the
online primal-dual LP-solving framework. To
remedy this, we use a ``round-or-separate'' approach:
we use the rounding algorithm itself as an approximate separation
oracle. Indeed, we give an approximate separation oracle that given an
arbitrary solution promises to either round it to a feasible integer
solution with a similar objective value, or else to output a constraint
violated by the
solution.  
This approach to blur the line between rounding and separation may prove
useful for other problems in online algorithms, where we are faced with
an exponential number of new constraints at each step. 
Furthermore, we give a new exponential-clock based sampling procedure
that finds constraints violated by a large margin with high
probability. Our analysis extends a technique of
\cite{vondrak2007submodularity} that approximates independent random
sampling by a continuous process. We use a potential function argument
to conclude that we only need a polynomial number of calls to the
separation oracle in the worst case. 

\subsection{Related Work}

\paragraph{Submodular Cover.} Wolsey~\cite{Wolsey1982} showed that the natural greedy algorithm gives
a $1 + \ln (\fmax / \fmin)$ approximation for \ssc; this guarantee is
tight unless $\P = \NP$ even for the special case of set
cover~\cite{Dinur:2014:AAP:2591796.2591884}. Fujito~\cite{fujito2000approximation}
gave a dual greedy algorithm that generalizes the
$F-$approximation  for set cover~\cite{doi:10.1137/0211045}. Neither
approach seems portable to the online setting (even for set cover),
since changing $f$ to $f'$ changes the greedy/primal-dual solution in
non-smooth ways. One can also solve \ssc via repeated sequential use of submodular maximization
 subject to a cardinality/knapsack constraint (e.g. \cite{sviridenko2004note} or the survey \cite{buchbinder2017submodular}), but once again it is not clear how to do this online when the function $f$ changes. 

\ssc has been used in many applications to resource allocation and placement
problems, 
by showing that the coverage function is monotone and submodular,
and then applying Wolsey's greedy algorithm as a black box. 
We port these applications to the online setting where coverage
requirements change with time. 
E.g., in selecting influential nodes to disseminate information in social networks \cite{Goyal2013, loukides2016limiting, tong2017positive, IZUMI20102773}, exploration for robotics planning problems \cite{krause2008robust, jorgensen2017risk, beinhofer2013robust}, placing sensors 
\cite{WU201553, rahimian2015detection, zheng2017trading, 7504484}, and
other physical resource allocation objectives \cite{6996018, 7798894,
  tzoumas2016minimal}. There has been much recent interest in \ssc from the
networking community, as \ssc models network function placement tasks 
\cite{Andreev:2009:SSL:1644015.1644031, Lee:2013:FPC:2523616.2525960,
  kortsarz2015approximating, LUKOVSZKI2018159, chen2018virtual}. E.g., 
\cite{LUKOVSZKI2018159} want to place middleboxes in a
network incrementally, and point out that 
avoiding closing extant boxes is a huge boon in practice. We extend this
application to the case where 
the coverage function $f$ changes.

There are previous definitions of online submodular cover, but
similarity 
to our work is mainly in the name.  In one 
variant inspired by regret-minimization in online
learning~\cite{NIPS2008_3569, NIPS2011_4465},
a submodular function $f^{(t)}$ arrives at each time step $t$, and we
must output a permutation of the universe elements that minimizes the
(average or maximum) cover time. 
The approach compares the algorithm's objective value to the single best
\textit{fixed} cover in hindsight. Hence, there is no notion of
time-monotonicity, and of element selection being irrevocable.
Another distantly related line of work (\cite{golovin2011adaptive,
  Deshpande:2016:AAS:2930058.2876506, hellerstein2018revisiting,
  grammel2016scenario, agarwal2019stochastic}) explores a stochastic
variant of \ssc, where there is a fixed submodular function $f$, but
each element of the universe is active only with a certain
probability. The goal is to output an ordering of the elements
minimizing the expected cover time.

\paragraph{Online Coverage Problems.} 
In online set cover,
``items'' arrive online one by one; when
an item arrives we find out which sets cover this
item. \cite{alon2003online,buchbinder2009online} give an online algorithm with a
competitive ratio of $O(\log n \log T)$, where there are $n$ sets and
$T$ items. This is based on solving the LP relaxation via the online
primal-dual framework, followed by randomized rounding.
Our \ossc problem generalizes the online set cover problem;
if we identify $N$ with the sets, and $f^{(t)}$ is the coverage function
induced by the set system restricted to the first $t$ items.
\cite{buchbinder2009online,gupta2014approximating} extend this result to
a wide class of covering and packing LPs/IPs, with competitiveness
$O(\log n \log T)$ for solving $\{0,1\}$ covering IPs with $n$ variables
and $T$ constraints, and better results for row/column sparse instances.

\cite{gupta2014changing} 
maintain a matroid base (or spanning set) under
changing costs. Given a fixed matroid $\mathcal{M}$, each matroid
element is given holding and and acquisition costs at each time
$t$, after which we must output a matroid base (spanning set). We want to minimize the
sum of acquisition plus holding costs over time. In the spanning set case, when holding costs are
zero, acquisition costs are fixed and the matroid is additionally changing over time, this is the very special case of our \ossc problem where the
submodular functions are matroid rank functions.
Recall that matroid rank functions are a very restricted class of
submodular functions: for instance, we can efficiently find a min-weight
matroid base, while no efficient algorithms exist for finding a minimum
weight submodular cover unless $\P = \NP$.  
 \cite{buchbinder2014competitive}  study the online matroid caching problem, 
where one must maintain an independent set in a matroid subject to the constraints that the set must contain certain elements at specified points in time.

\subsection{Notation and Preliminaries.}
\label{subsec:notation}

A set function $f: 2^{N} \rightarrow \mathbb{R}^+$ 
is \textit{submodular} if $f(A \cap B) + f(A \cup B) \leq f(A) + f(B)$
for any $A, B \subseteq N$.  It is \textit{monotone} if $f(A) \leq f(B)$
for all $A \subseteq B \subseteq N$. We assume access to a \emph{value
  oracle} for $f$ that computes $f(T)$ given $T \subseteq N$. The
\textit{contraction} of $f: 2^N \rightarrow \mathbb{R}^+$
onto $N\setminus T$ is defined as 
 $f_T(S) = f(S \mid T) := f(S \cup T) - f(T)$. 
If $f$ is submodular then $f_T$ is also submodular for any
$T \subseteq N$. We use the following notation: 
\begin{align}
\fmax &:= \displaystyle \max \left\{ f^{(t)}(j \mid S) \Big \vert 
j \in N, \ t \in [T], \ S \subseteq N \right\} \label{line:fmaxdef}\\
\fmin &:= \displaystyle \min \left\{ f^{(t)}(j \mid S) \Big \vert 
j \in N, \ t \in [T], \ S \subseteq N, \
f^{(t)}(j \mid S) \neq 0 \right\} \label{line:fmindef}
\end{align}
Note that if $f^{(t)} = \sum_{s \leq t} g_s$ is the prefix sum of monotone submodular
functions arriving online, then $\fmax, \fmin$ are at most $T$ times
maximum and minimum marginals of these $g_s$ functions. Also we let $\cmax$ and $\cmin$ denote the largest and smallest costs of elements respectively.

The well-known \textit{multilinear extension} of $f$ is 
defined as:
\begin{align}
F(x) := \expectarg{f(S)}{S \sim x} = \sum_{S \subseteq N} \prod_{i \in S} x_i \prod_{j \not \in S} (1 - x_j) f(S) \label{line:fmultilinear}
\end{align}
where $S \sim x$ means we add each element $j \in N$ independently to
$S$ with probability $x_j$. We can naturally extend definition \eqref{line:fmultilinear} to contractions of functions: $F_T(x) := \expectarg{f_T(S)}{S \sim x}$.

\emph{Time-monotonicity} generalizes the notion of an insertion-only arrival model for coverage functions, and is consistent with the (insertion-only) formulation of online set cover as defined in \cite{alon2003online}.
\begin{definition}[Time-monotone]
	A sequence of monotone submodular functions $f^{(0)}, f^{(1)},
        f^{(2)}, \ldots, f^{(T)}$ is \textit{time-monotone} if for all
        $t$ and $S$, we have $f^{(t)}(S) = f^{(t)}(N) \implies f^{(t-1)}(S) = f^{(t-1)}(N)$.
\end{definition}
\subsection{Higher Order Monotonicity and Mutual Coverage.} To describe some of our results, we require some definitions which generalize the idea of monotonicity and submodularity.

Following the notation of \cite{DBLP:journals/mor/FoldesH05}, we define the \textit{derivative of a set function} $f$ as:
\[\frac{df}{dx} (S) = f(S \cup \{x\}) - f(S \backslash \{x\}).\]
We notate the $m$-th order derivative of $f$ with respect to the subset $A = \{i_1, \ldots i_m\}$ as $df/dA$, and this quantity has the following concise expression:
\[\frac{df}{dA}(S) = \sum_{B \subseteq A} (-1)^{|B|} f((S \cup A) \backslash B).\]

\begin{restatable}[$m$-increasing \cite{DBLP:journals/mor/FoldesH05}]{definition}{HOmonotone}
	\label{def:HOmonotone}
	We say that a set function is \textit{$m$-increasing} if all its $m$-th order derivatives are nonnegative, and \textit{$m$-decreasing} if they are nonpositive. We denote by $\mathcal{D}_m^+$ and $\mathcal{D}_m^-$ the classes of $m$-increasing and $m$-decreasing functions respectively.
\end{restatable}

Note that $\mathcal{D}_1^+$ is the class of monotone set functions, and $\mathcal{D}_2^-$ is the class of submodular set functions. When $f$ is the joint entropy set function, the $m$-th derivative above is also known as the \textit{interaction information}, which generalizes the usual mutual information for two sets of variables, to $m$ sets of variables. We now focus on the usual $m=2$ case, but relax $f$ to be a general nonnegative, monotone, submodular function.

\begin{definition}[Mutual Coverage]
  \label{def:mutualCoverage}
  The \textit{mutual coverage} and \emph{conditional mutual
    coverage} with respect to a set function
  $f: 2^N \rightarrow \mathbb{R}^+$ are 
  defined as:
  \begin{align}
    \mutcov(A ; B) &:= f(A) + f(B) - f(A \cup B) \label{eq:mutcov} \\
    \mutcov(A ; B \mid C) &:= f_C(A) + f_C(B) - f_C(A \cup B) \label{eq:cond-mutcov}
  \end{align}
\end{definition}
Independently, \cite{DBLP:conf/alt/IyerKBA21} defined and studied the same quantity under the slightly different name \textit{submodular mutual information}.

We write $\mutcov^{(A)}(B) := \mutcov(A ; B)$ and
$\mutcov^{(A\mid C)}(B) := I(A ; B \mid C)$ if we want to view mutual
coverage as a function of $B$ for fixed $A$ and $C$. We may think of
$\mutcov^{(A \mid C)}(B)$ intuitively as the amount of coverage $B$
takes away from the coverage of $A$ given that $C$ was already chosen
(or vice-versa, since the definition is symmetric in $A$ and $B$).

\begin{fact}[Chain Rule]
	\label{fact:chainRule}
  Mutual coverage respects the identity:
  \begin{align*}
    &\mutcov(A ; B_1 \cup B_2 \mid C) = \mutcov(A; B_1 \mid C) + \mutcov(A; B_2 \mid C \cup B_1)
  \end{align*}	
\end{fact}
This neatly generalizes the chain rule for mutual information.
The expression $\mutcov^{(A \mid C)}(B)$ is monotone in $B$ by the
submodularity of $f$, but not submodular in general (see \cref{section:appendixStrongSubmod}).

We will soon show improved algorithms for the class $\mathcal{D}_3^+$, that is the class of $3$-increasing set functions. The following derivation gives some intuition for these functions.
\begin{align*}
	&\frac{df}{d\{x,y,z\}}(S) = \sum_{B \subseteq \{x, y, z\}} (-1)^{|B|} f((S \cup \{x,y,z\}) \backslash B)  \\
	&= f(x \mid S) - f(x \mid S \cup \{y\}) - ( f(x \mid S \cup \{z\}) - f(x \mid S \cup \{y, z\}) ) = \mutcov(x,y \mid S) - \mutcov(x,y \mid S \cup \{z\}).
\end{align*}
Thus a function $f$ in contained in $\mathcal{D}_3^+$ if and only if mutual coverage decreases after conditioning, in other words mutual coverage is submodular in each of its arguments.

\section{An $O\left(\ln n \ln \left( T \cdot f(N) / \fmin \right) \right)$ competitive algorithm}
\label{subsec:slowgreedy}

In this section, we show how to get our
first 
algorithm for \ossc: we use the approach of solving covering LPs online
to maintain a competitive fractional solution $x_t$ at every time $t$,
with the special property that $x_1, x_2, \ldots, x_T$ are monotonically
increasing, i.e., $x_t \geq x_{t-1}$. 
To get an integer solution online, we use randomized rounding with the
method of alterations:
we round the increments in variables at each step, and  make corrections in the case of failures by greedily selecting elements until our solution is feasible. Our algorithm always outputs a feasible solution; the challenge is to bound its expected cost: the crucial technical component is a relationship
between the submodular cover LP and the multilinear relaxation, which allows us to
show that a logarithmic number of rounds suffice.

\subsection{An LP for Submodular Cover.}

How can we frame submodular cover as a linear program? We use the well known formulation used by Wolsey \cite{Wolsey1982} (see also \cite{Nemhauser1978, Fisher1978}), but first we define some notation for convenience. For $S\subseteq N$, define $f^{\#S}: [0,1]^{|N|}
\rightarrow \mathbb{R}^+$ as:
\begin{align}
\ts f^{\#S}(x) := f(S) + \sum_{j \in N} f_S(j) x_j \label{line:fsharp}
\end{align}
We define the \textit{covering extension} of $f$ as:
\begin{align}
f^*(x) := \min_{S \subseteq N} f^{\#S}(x) \label{line:fstar}
\end{align}
When $f$ is submodular, the function $f^*$ is a continuous
extensions of $f$: i.e., if $\chi_T$ is the characteristic vector of a
set $T \subseteq N$, then $f^{*}(\chi_T) = F(\chi_T) = f(T)$. We will also sometimes need to extend definitions \eqref{line:fsharp} and \eqref{line:fstar} to contractions of functions.
\begin{align*}
f^{\#S}_T(x) &:= f_T(S) + \sum_{j \in N} f_{S \cup T}(j) x_j \\
f^*_T(x) &:= \min_{S \subseteq N} f_T^{\#S}(x)
\end{align*}
 
Now we can write the submodular cover LP as:
\begin{gather*}
\boxed{
\begin{array}{lll}
& \displaystyle \min  \sum_j c(j)\, x_j & \\
\text{s.t.\quad} &  \displaystyle \sum_{j \in N} f_S(j)\, x_j \geq f_S (N) & \forall S \subseteq N \\
&1 \geq x_j \geq 0 & \forall j
\end{array}
}
\intertext{or in our notation:}
\boxed{
	\begin{array}{lll}
	& \displaystyle \min  \sum_j c(j)\, x_j & \\
	\text{s.t.\quad} &  f^*(x) \geq f(N) \\
	&1 \geq x_j \geq 0 & \forall j
	\end{array}
}
\label{eq:submodCoverLP}
\tag{$P$}
\end{gather*}
The integer solutions of this LP are precisely the solutions to \ssc:
\begin{lemma}[Proposition 2 of \cite{Wolsey1982}]
	\label{lem:integerpts}
	A set $T$ has $f(T) = f(N)$ if and only if $\chi_T$, the
	characteristic vector of $T$, is a feasible integer solution to
	\eqref{eq:submodCoverLP}.
\end{lemma}

Let $(P^{(t)})$ refer to the program \eqref{eq:submodCoverLP} for $f^{(t)}$. We would like to argue that we can maintain one large LP for
\ossc. At time $t$, we feed in the constraints of $(P^{(t)})$, and
update the solution accordingly. We use the following online LP solver from~\cite[Theorem 4.2]{buchbinder2009online}:
\begin{theorem}
	\label{thm:buchbinder2009}
	For the setting where the rows of a covering LP $\min\{ c^\intercal x \mid Ax \geq b \}$ with $A \in \R_+^{m\times n}$, $b \in \R^m_+$, $c \in \R^n_+$ arrive online, there is an algorithm that achieves a competitive ratio of $O(\ln n)$. Furthermore, the sequence of solutions produced $x_1, x_2, \ldots, x_T$ is monotonically increasing.
\end{theorem}

Each of the linear programs $(P^{(t)})$ is indeed a covering LP with box
constraints that is a feasible relaxation for the function
$f^{(t)}$. However it is not \emph{a priori} clear that the constraints for
the linear program $(P^{(t')})$ for some $t' < t$ are also valid for
the linear program $(P^{(t)})$. We prove a lemma which circumvents this issue. While we do
not guarantee that for all $t' < t$ the constraints of $(P^{(t')})$ will
be valid for the LP $(P^{(t)})$ itself, we show that they are valid for the
\textit{integer} solutions of $(P^{(t)})$. This suffices to ensure that
the union of all constraints seen so far is a relaxation for \ossc. 
\begin{lemma}
Any feasible \underline{integer} solution to $(P^{(t+1)})$ is a feasible \underline{integer} solution to $(P^{(t)})$.
\end{lemma}

\begin{proof}
Suppose $x$ is a feasible integer solution to $(P^{(t+1)})$. Then by
\cref{lem:integerpts}, $x$ is the characteristic vector of a set $S$
such that $f^{(t+1)}(S) = f^{(t+1)}(N)$. By the time-monotonicity of the
sequence $f^{(0)}, f^{(1)}, \ldots, f^{(T)}$, it holds that $f^{(t)}(S)
= f^{(t)}(N)$ as well. Now again using \cref{lem:integerpts}, $x$ is
feasible to $(P^{(t)})$. 
\end{proof}
\subsection{Rounding Fractional Solutions.}

The last remaining step in our plan is to understand the behavior of
randomized rounding for \eqref{eq:submodCoverLP}. Since the multilinear
extension $G$ is defined as the value of $g$ on a set obtained by randomized
rounding, we can use the following relationship between $G$ and $g^*$~\cite[Lemma 3.8]{vondrak2007submodularity}:
\begin{lemma}
	\label{lem:multilinFstar}
	For any monotone, nonnegative set function $g$, it holds that $G(x) \geq (1-e^{-1}) g^*(x)$.
      \end{lemma}
Observe that \cref{lem:multilinFstar} does not require that $g$ be
submodular, only that it be monotone and nonnegative. (There exist functions for which the bound is tight. If $g$ is also
		submodular, we also get the bound $G(x) \leq
		g^*(x)$~\cite[Lemma~3.7]{vondrak2007submodularity}.)
The next lemma helps us analyze random rounding:
\begin{lemma}[Rounding Lemma]
	\label{lem:roundingLemma}
	Let $g$ be a monotone, nonnegative set function such that $g(\emptyset) = 0$, and let $k$ be a positive integer. Let $R_1, \ldots R_k$ be a sequence of random sets and denote $R_{1:\ell} := \bigcup_{i \in [\ell]} R_i$. Suppose the following condition holds for all $\ell \in [k]$ and all $R \subseteq N$:
 	\begin{align}
 	\expectarg{g(R_\ell \mid R_{1:\ell-1}) \mid R_{1:\ell-1} = R}{R_{\ell}} \geq (1 - \gamma) g(N \mid R) \label{line:roundingLemma_roundingassump}
 	\end{align}
 	 Then:
	\begin{align}
	\expectarg{g(N \mid R_{1:k})}{R_{1:k}} \leq \gamma^k \cdot g(N) \label{line:roundingLemma_lemmaConc}
	\end{align}
\end{lemma}
\begin{proof}
	By definition:
	\begin{align}
	\expectarg{g(N \mid R_{1:\ell})}{R_{1:\ell}} &= \expectarg{g(N) - g(R_{1:\ell-1}) - \expectarg{g(R_{1:\ell}) - g(R_{1:\ell-1})}{R_{\ell}}}{R_{1:\ell-1}} \notag\\
	& \leq \expectarg{g(N \mid R_{1:\ell-1}) - (1 - \gamma)(g(N \mid R_{1:\ell-1}))}{R_{1:\ell-1}} \label{line:roundingLemma_roundingworks}\\
	&= \gamma \cdot\expectarg{g(N) - g(R_{1:\ell-1})}{R_{1:\ell-1}} \notag
	\end{align}
	Note that \eqref{line:roundingLemma_roundingworks} holds by assumption \eqref{line:roundingLemma_roundingassump}. Claim \eqref{line:roundingLemma_lemmaConc} holds by induction on $\ell$ and the fact that $g(\emptyset) = 0$.
\end{proof}
\begin{lemma}
	\label{lem:fNfminround}
	Let $x \in [0,1]^n$ be a feasible solution to \eqref{eq:submodCoverLP}. Let $R$ be a set obtained by performing $k$ rounds of randomized rounding according to $x$. Then: $\expectarg{f_R(N)}{R} \leq e^{-k} f(N)$.
\end{lemma}
\begin{proof}
	Since $x$ is feasible to $(P)$, for all $S\subseteq N$, we have $f_S^*(x) \geq f_S(N)$.	By \cref{lem:multilinFstar} applied to $f_S$, it holds that $F_S(x) \geq (1 - e^{-1}) f_S^*(x)$. Together these imply: 
	\begin{align}
	F_S(x) \geq (1 - e^{-1}) f_S(N)
	\end{align}
	Thus $x$, $f$ and $R$ together satisfy the conditions of the \nameref{lem:roundingLemma} with $\gamma = e^{-1}$, and the claim follows directly.
\end{proof}

\subsection{The Analysis.}

To summarize, our first algorithm for \ossc is the following: maintain a
fractional solution $x$, and two sets $R, G \subseteq N$ initially
empty. At time $t$, feed the batch of constraints of $(P^{(t)})$ to
\cite{buchbinder2009online} and update $x$ accordingly. Set $k := \ln(t^2 \cdot f(N) / \fmin)$, and for every $j\in N$, add $j$ to $R$ with probability $\Delta_j$, where $\Delta_j$ is calculated such that the total probability that $j$ gets added to $R$ over the course of the algorithm is $1 - (1 - x_j)^k$. In a subsequent greedy phase, iteratively select the cheapest element $j \in N$ with non-zero marginal
value and add it to the set $G$ until no such elements remain. Return
$R \cup G$. 

Given our tools from the previous section, we can now
analyze the expected cost.

	\begin{theorem}
		\label{thm:slowOssc}
		There exists a randomized algorithm for the \ossc problem
		which guarantees that for each $T$, the expected cost of solution
		$S_T$ is within $O(\ln n \cdot \ln (T\cdot f(N) / \fmin))$ of the
		optimal submodular cover solution for $f^{(T)}$.
	\end{theorem}
	\begin{proof}
		By construction, the algorithm always produces a feasible solution to the submodular cover problem. Thus it suffices to relate the cost of the output to the cost of the optimal solution, $c(\opt)$.
		
		To bound the cost of the sampling phase, remark that at time $t$, by construction we can view $R$ as the result of $k$ rounds of independent random rounding. The expected cost after one round of rounding is precisely $c(x)$. By linearity of expectation, the total expected cost of $R$ is $\ln \left(t^2\cdot f(N) / \fmin \right)\cdot c(x)$. 

		Next we bound the cost of the greedy phase. 
		At time $t$, the solution $x$ is feasible to $(P^{(t)})$. By \cref{lem:fNfminround} with parameter $k= \ln \left(t^2\cdot f(N) / \fmin\right)$:
		\begin{align*}
		\expectarg{f_R(N)}{R} \leq t^{-2} \left(\frac{\fmin}{f(N)}\right) f(N) = t^{-2} \fmin
		\end{align*}
		This gives a bound on the expected cost of alterations at time $t$. Let $G_t$ denote the elements added in the greedy phase at time $t$:
		\begin{align}
		\expectarg{c(G_{t})}{R} &\leq c(\opt) \cdot \expectarg{|G_t|}{R} \label{line:fNround_cheaperthanopt}\\
		&\leq t^{-2} c(\opt) \label{line:fNround_expectedgreedy}
		\end{align}
		\eqref{line:fNround_cheaperthanopt} comes from the fact that each element must be cheaper than $c(\opt)$, the second step \eqref{line:fNround_expectedgreedy} from the fact that for every element $h \in G$, both $f_{R}(h) \geq \fmin$ and $\expectarg{f_R(N)}{R} \leq t^{-2}\fmin$. By linearity of expectation, the expected cost of cumulative alterations over all time steps $t \in [T]$ is $O(c(\opt))$. In conclusion, the final solution has expected cost at most $ \ln \left(T\cdot f(N) / \fmin \right) c(x) + O(c(\opt))$. By \hyperref[thm:buchbinder2009]{Theorem 2.1}, $c(x) = O(\ln n) \cdot c(\opt)$, hence the algorithm outputs a solution with competitive ratio $O(\ln n \ln \left(T \cdot f(N) / \fmin \right)) \cdot  c(\opt)$.
	\end{proof}

In the next section, we show how to replace the $f(N)$ term in the
competitiveness by an $\fmax$ term, which is often much smaller. In
order to do the finer analysis, we will also show how to use the ideas
of mutual coverage. We then make these algorithms efficient in the subsequent
section.

\section{An $O\left(\ln n \ln \left( T \cdot \fmax / \fmin \right) \right)$ competitive algorithm for $\mathcal{D}_3^+$}

\label{subsec:betterround}
In this section, we show how to recover the finer $O\left(\ln n \ln \left( T \cdot \fmax / \fmin \right)\right)$ approximation for the special case of $3$-increasing functions. 
The main algorithmic change is to perform $\ln \left( T \cdot \fmax / \fmin \right)$ rounds of rounding instead of $\ln \left( T \cdot f(N) / \fmin \right)$, and then compensate with more greedy alterations. The intuition behind the analysis is a fractional charging scheme. Using mutual coverage as a notion of contribution, we argue that:
\begin{enumerate}[(i)]
	\item every greedily chosen element is cheaper than the elements in $\opt$ to which it contributes.
	\item every greedily chosen element contributes at least a minimum amount overall.
	\item for every element in $\opt$, the total contribution it receives from greedily chosen elements is bounded above.
\end{enumerate}
We show that these statements together imply a good bound on the
expected cost of alterations. However, this proof strategy requires that
the mutual coverage between pairs of elements decrease sufficiently
after randomized rounding. In \cref{section:appendixStrongSubmod}, we
show that the constraints of the LP \eqref{eq:submodCoverLP} do
not necessarily guarantee that this is the case;
hence we introduce additional constraints that enforce the desired property.

\subsection{A Stronger Formulation.}

Recall our definition of \hyperref[def:mutualCoverage]{mutual coverage}
$\mathcal{I}_f(A;B) = f(A)+f(B)-f(A\cup B)$ from \cref{subsec:notation}. We begin by strengthening LP \eqref{eq:submodCoverLP} with additional constraints as follows:
\begin{align*}
\boxed{
\begin{array}{lll}
\displaystyle \min \sum_j c(j) x_j \quad \text{s.t.} &\\
\displaystyle \sum_{j \in N} \mutcov^{(v\mid S)}(j)  x_j \geq \mutcov^{(v\mid S)}(N) & \forall v \in N, \forall S\subseteq N & (*)\\
1 \geq x_j \geq 0 & \forall j 
\end{array}
}
\label{eq:newRelaxation}
\tag{$Q$}
\end{align*}
Note that by assumption, that $f \in $, the function $\mutcov^{(\{v\})}$ is submodular, and hence the program above is the intersection of several submodular cover LPs (one for each $v$).

To see that the new constraints imply the constraints of \eqref{eq:submodCoverLP}, consider any set $S$. Order the elements of $N$ arbitrarily $v_1, v_2 \ldots, v_n$ and define $S_i := S \cup \{v_1, \ldots v_{i-1}\}$ with $S_1 := S$. If we sum the (\hyperref[eq:newRelaxation]{$*$}) constraints corresponding to all pairs $(v_i, S_i)$, we obtain:
\begin{align}
\sum_{i \in [n]} \sum_{j \in N} \mutcov(v_i; j \mid S_i) \cdot x_j &\geq \sum_{i \in [n]} \mutcov(v_i; N \mid S_i) \notag\\
\Rightarrow \sum_{j \in N} \mutcov(N; j \mid S) \cdot x_j &\geq \mutcov(N; N \mid S) \label{line:usingChainRuleMutCov} \\
\Rightarrow \sum_{j \in N} f_S(j) \cdot x_j &\geq f_S(N)\notag
\end{align}
where we used the \nameref{fact:chainRule} for mutual coverage for step \eqref{line:usingChainRuleMutCov}.

In order to treat our problem as an online covering LP and use \cite{buchbinder2009online} like we did before, we again need to make sure the new constraints do not eliminate any integer solutions.
\begin{lemma}
	\label{lem:fancyrelaxation}
	\eqref{eq:newRelaxation} is a relaxation of \ssc.
\end{lemma}
\begin{proof}
	Consider a feasible solution $T$ to \ssc and its corresponding characteristic vector $\chi_T$. Order the elements of $T$ arbitrarily $j_1, j_2 \ldots, j_k$ and define $S_i := S \cup \{j_1, \ldots j_{i-1}\}$ with $S_1 := S$. Then for any $S \subseteq N$ and $v \in N - S$:
	\begin{align}
	\sum_{j \in N } \mutcov(v, j \mid S) (\chi_T)_j &= \sum_{j \in T} \mutcov(v, j \mid S) \geq \sum_{i =1}^k \mutcov(v, j_i \mid S_i) \label{line:submodmutcov}\\
	&= \mutcov(v, T \mid S) \label{line:cr_again} = \mutcov(v, N \mid S). 
	\end{align}
		Above, inequality \eqref{line:submodmutcov} used the $3$-increasing property, which implies the submodularity of mutual coverage. Step \eqref{line:cr_again} used the \nameref{fact:chainRule}.
	\end{proof}

As we will soon see, the purpose of the new (\hyperref[eq:newRelaxation]{$*$}) constraints is to ensure that, just like $f_R(N)$, the expected remaining mutual coverage decreases geometrically with the number of rounds.

\subsection{The Improved Analysis.}

Our second algorithm for \ossc resembles the first; it is the analysis
that differs. Indeed, at time $t$, feed the batch of constraints of
$(Q^{(t)})$ (defined as \eqref{eq:newRelaxation} for the function $f^{(t)}$) to \cite{buchbinder2009online} and update the fractional
solution $x$ accordingly. In the rounding phase, perform
$k = \ln(t^2 \cdot \fmax / \fmin)$ rounds of randomized rounding using the appropriate probability vector $\Delta \in [0,1]^n$ (recall $\Delta_j$ is calculated such that the total probability that $j$ gets added to $R$ over the course of the algorithm is $1 - (1 - x_j)^k$) and add the new elements to the set $R$. In a subsequent greedy phase, iteratively add the cheapest
element $j \in N$ with non-zero marginal value to set $G$, until none
remain. Return $R \cup G$.

It is clear that the algorithm is correct, so it remains to bound the
expected cost. The main idea is to keep track of the amount of coverage
that the $t^{th}$ greedy element in the augmentation procedure contributes
to the coverage achieved by the $s^{th}$ element of the optimal
solution; the (conditional) mutual coverage between the two is the
quantity of interest here.

\begin{theorem}
	\label{thm:slowOnlineSSC2}
	There exists a randomized algorithm for the \ossc problem when $f \in \mathcal{D}_3^+$
	which guarantees that for each $T$, the expected cost of solution
	$S_T$ is within $O(\ln n \cdot \ln (T\cdot \fmax / \fmin))$ of the
	optimal submodular cover solution for $f^{(T)}$.
\end{theorem}

\begin{proof}
	The algorithm always outputs a feasible solution, and it suffices to relate the cost of $R \cup G$ to the cost of the optimal solution. As before, at every time $t \in [T]$, the probability that $j$ is in $R$ is $1 - (1 - x_j)^k$ and we may view $R$ as the result of $k$ rounds of randomized rounding. Thus the expected cost of $R$ is $\ln (t^2\cdot \fmax / \fmin) \cdot c(x)$. In the remainder, we analyze the greedy phase.
	
	At time $t$, the fractional solution $x$ is feasible to $Q^{(t)}$. If $S_{t-1}$ is the solution the algorithm produced for time $t-1$, define $f := f^{(t)}_{S_{t-1}}$. Then $x$ is also feasible to \eqref{eq:newRelaxation} for the contracted function $f$. By \cref{lem:fNfminround} with $k=\ln(t^2 \cdot \fmax / \fmin)$: 
	\begin{align}
	\expectarg{\mutcov (v; N \mid R)}{R} &\leq t^{-2}\left(\frac{\fmin}{\fmax}\right) \mutcov (v; N) \\
	&= t^{-2}\left(\frac{\fmin}{\fmax}\right) f(v) \label{line:onlineSSC2_mutCovRoundingWorks}
	\end{align}
	Let $G_1,\ldots, G_p$ denote the elements added by the greedy phase in increasing order of cost, and let $G_{1:t} = \{G_1,\ldots, G_t\}$ be the prefix of the first $t$ elements. Also order the elements of $\opt$ arbitrarily $O_1, \ldots, O_q$ and similarly let $O_{1:s} = \{O_1, \ldots, O_s\}$. For convenience, let $G_{1:0} = O_{1:0} = \emptyset$.
	
	We say an element $G_{t}$ \textit{contributes} to $O_{s}$ if conditioned on having selected $G_{1:t-1}$, and $O_{1:s-1}$, selecting $G_t$ changes the marginal value of $O_s$. Let $C(t,s)$ be the contribution of $G_{t}$ to $O_{s}$, or the amount by which this marginal value changes. Formally:
	\begin{align*}
	C(t,s) &:= \mutcov(G_{t}; O_{s} \mid G_{1:t-1} \cup O_{1:s-1} \cup R)
	\end{align*}
	
		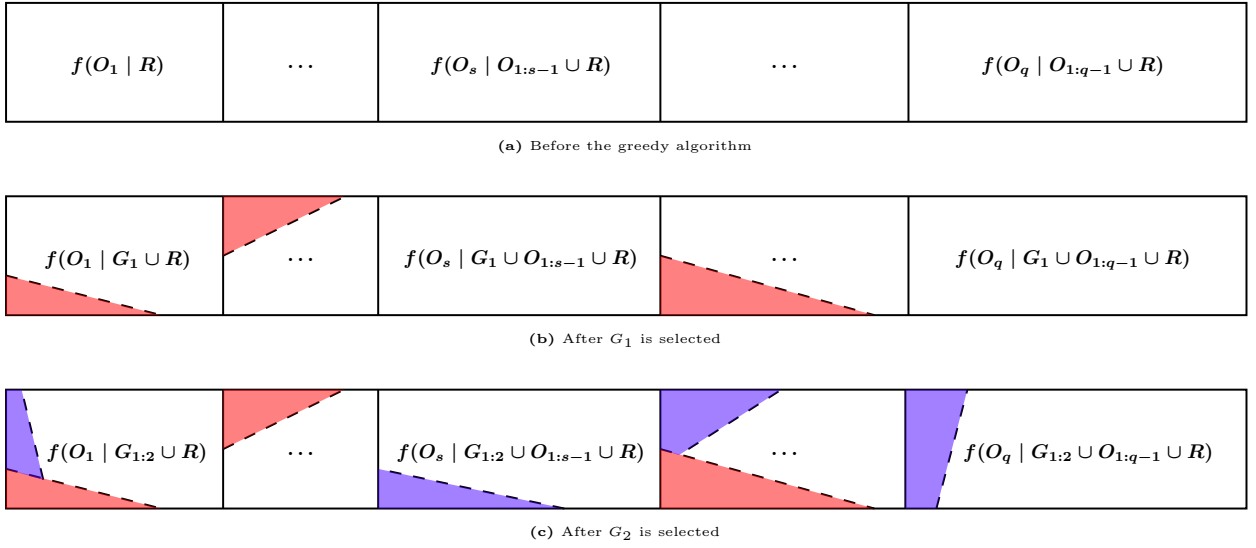
\begin{figure*}
	\centering
	
	\tikzset{every picture/.style={line width=0.75pt}} %set default line width to 0.75pt      
	\begin{subfigure}{\textwidth}
		\begin{tikzpicture}[x=1.173pt,y=0.75pt,yscale=-1,xscale=1,every node/.style={scale=0.75}]
		\centering
		\draw   (110,20) -- (510,20) -- (510,80) -- (110,80) -- cycle ;
		\draw    (180,20) -- (180,80) ;
		\draw    (230,20) -- (230,80) ;
		\draw    (321,20) -- (321,80) ;
		\draw    (401,20) -- (401,80) ;
		%			\draw (541,52) node  [align=center] {$\displaystyle f_R(N)$};
		\draw (146,52) node  [align=center] {$\boldsymbol{\displaystyle f(O_{1} \mid R)}$};
		\draw (205,52) node  [align=center] {$\boldsymbol{\ldots}$};
		\draw (276,52) node  [align=center] {$\boldsymbol{\displaystyle f( O_{s} \mid O_{1:s-1} \cup R)}$};
		\draw (361,52) node  [align=center] {$\boldsymbol{\ldots}$};
		\draw (454,52) node  [align=center] {$\boldsymbol{\displaystyle f( O_{q} \mid O_{1:q-1} \cup R)}$};
		\end{tikzpicture}
		\caption{Before the greedy algorithm}
		\label{fig:beforeGreedy}
	\end{subfigure}
	
	\vspace{0.2in}
	
	\begin{subfigure}{\textwidth}
		\begin{tikzpicture}[x=1.173pt,y=0.75pt,yscale=-1,xscale=1,every node/.style={scale=0.75}]
		\centering
		\draw   (110,150) -- (510,150) -- (510,210) -- (110,210) -- cycle ;
		\draw    (180,150) -- (180,210) ;
		\draw    (230,150) -- (230,210) ;
		\draw    (321,150) -- (321,210) ;
		\draw    (401,150) -- (401,210) ;
		\draw  [dash pattern={on 4.5pt off 4.5pt}]  (110,190) -- (160,210) ;
		\draw  [dash pattern={on 4.5pt off 4.5pt}]  (180,180) -- (219.5,150) ;
		\draw  [dash pattern={on 4.5pt off 4.5pt}]  (321,180) -- (390,210) ;
		\draw  [draw opacity=0][fill={rgb, 255:red, 255; green, 0; blue, 0 }  ,fill opacity=0.5 ] (110,190) -- (160,210) -- (110,210) -- cycle ;
		\draw  [draw opacity=0][fill={rgb, 255:red, 255; green, 0; blue, 0 }  ,fill opacity=0.5 ] (180,180) -- (219,150) -- (180,150) -- cycle ;
		\draw  [draw opacity=0][fill={rgb, 255:red, 255; green, 0; blue, 0 }  ,fill opacity=0.5 ] (321,180) -- (390,210) -- (321,210) -- cycle ;
		
		\draw (146,182) node  [align=center] {$\boldsymbol{\displaystyle f(O_{1} \mid G_1 \cup R)}$};
		\draw (205,182) node  [align=center] {$\boldsymbol{\ldots}$};
		\draw (276,182) node  [align=center] {$\boldsymbol{\displaystyle f(O_{s} \mid G_1 \cup O_{1:s-1} \cup R)}$};
		\draw (361,182) node  [align=center] {$\boldsymbol{\ldots}$};
		\draw (454,182) node  [align=center] {$\boldsymbol{\displaystyle f(O_{q} \mid G_1 \cup O_{1:q-1} \cup R)}$};
		\end{tikzpicture}
		\caption{After $G_1$ is selected}
		\label{fig:afterGreedy1}
	\end{subfigure}
	
	\vspace{0.2in}
	
	\begin{subfigure}{\textwidth}
		\begin{tikzpicture}[x=1.173pt,y=0.75pt,yscale=-1,xscale=1,every node/.style={scale=0.75}]
		\centering
		\draw   (110,280) -- (510,280) -- (510,340) -- (110,340) -- cycle ;
		\draw    (180,280) -- (180,340) ;
		\draw    (230,280) -- (230,340) ;
		\draw    (321,280) -- (321,340) ;
		\draw    (400,280) -- (400,340) ;
		\draw  [dash pattern={on 4.5pt off 4.5pt}]  (110,320) -- (160,340) ;
		\draw  [dash pattern={on 4.5pt off 4.5pt}]  (180,310) -- (219,280) ;
		\draw  [dash pattern={on 4.5pt off 4.5pt}]  (321,310) -- (390,340) ;
		\draw  [dash pattern={on 4.5pt off 4.5pt}]  (122,325.5) -- (115,280) ;
		\draw  [dash pattern={on 4.5pt off 4.5pt}]  (290,340) -- (230,320) ;
		\draw  [dash pattern={on 4.5pt off 4.5pt}]  (330,310) -- (360,280) ;
		\draw  [dash pattern={on 4.5pt off 4.5pt}]  (410,340) -- (420,280) ;
		\draw  [draw opacity=0][fill={rgb, 255:red, 255; green, 0; blue, 0 }  ,fill opacity=0.5 ] (110,320) -- (160,340) -- (110,340) -- cycle ;
		\draw  [draw opacity=0][fill={rgb, 255:red, 255; green, 0; blue, 0 }  ,fill opacity=0.5 ] (180,310) -- (219,280) -- (180,280) -- cycle ;
		\draw  [draw opacity=0][fill={rgb, 255:red, 255; green, 0; blue, 0 }  ,fill opacity=0.5 ] (321,310) -- (390,340) -- (321,340) -- cycle ;
		\draw  [draw opacity=0][fill={rgb, 255:red, 70; green, 0; blue, 255 }  ,fill opacity=0.5 ] (230,320) -- (290,340) -- (230,340) -- cycle ;
		\draw  [draw opacity=0][fill={rgb, 255:red, 70; green, 0; blue, 255 }  ,fill opacity=0.5 ] (327,313) -- (321,310) -- (321,280) -- (360,280) -- cycle ;
		\draw  [draw opacity=0][fill={rgb, 255:red, 70; green, 0; blue, 255 }  ,fill opacity=0.5 ] (420,280) -- (400,280) -- (400,340) -- (410,340) -- cycle ;
		\draw  [draw opacity=0][fill={rgb, 255:red, 70; green, 0; blue, 255 }  ,fill opacity=0.5 ] (110,320) -- (122,325.5) -- (115,280)-- (110,280) -- cycle ;
		
		\draw (149,312) node  [align=center] {$\boldsymbol{\displaystyle f(O_{1} \mid G_{1:2} \cup R)}$};
		\draw (205,312) node  [align=center] {$\boldsymbol{\ldots}$};
		\draw (276,312) node  [align=center] {$\boldsymbol{\displaystyle f(O_{s} \mid G_{1:2} \cup O_{1:s-1}  \cup R)}$};
		\draw (361,312) node  [align=center] {$\boldsymbol{\ldots}$};
		\draw (459,312) node  [align=center] {$\boldsymbol{\displaystyle f(O_{q} \mid G_{1:2} \cup O_{1:q-1} \cup R)}$};
		\end{tikzpicture}
		\caption{After $G_2$ is selected}
		\label{fig:afterGreedy2}
	\end{subfigure}
	
	\vspace{0.1in}
	
	\caption{Illustration of the charging scheme. The white areas initially represent the marginal values of each element of $\opt$ conditioned on the elements that come before it in the ordering. The red areas represent the contribution of $G_1$ to the elements of $\opt$, or the amount by which the marginal value of each element in $\opt$ changes after conditioning on $G_1$. Likewise, the purple areas represent the contribution of $G_2$ to the elements of $\opt$. }
\end{figure*}
	
	We will need the three claims we alluded to in the opening of the section to proceed.
	\begin{claim}
		\label{claim:cheapcontribute}
		If $G_t$ contributes to $O_s$, i.e. $C(t,s) > 0$, then $c(G_{t}) \leq c(O_{s})$.
	\end{claim}	
	\begin{claim}
	\label{claim:fminLB}
	For any index $t \in [p]$ we have: $\sum_s C(t,s) \geq \fmin$.
	\end{claim}	
	\begin{claim}
	\label{claim:fmaxUB}
	For any index $s \in [q]$ we have: $\expectarg{ \sum_t C(t,s)}{R} \leq t^{-2}\fmin $.
\end{claim}
Before we prove these claims, let us show how they imply the statement of the theorem.
\begin{align}
\expect{\sum_t c(G_{t})} &= \expect{\sum_t \left(\sum_s \frac{C(t,s)}{\sum_{s'} C(t,s')} \right) c(G_{t})} \notag \\
&\leq \expect{\sum_s \left(\sum_t \frac{C(t,s)}{\sum_{s'} C(t,s')}\right) c(O_{s})} \label{line:hardround_claim1} \\
&\leq \displaystyle \sum_s \frac{\expect{\sum_t C(t,s)}}{\fmin} \cdot c(O_{s}) \label{line:hardround_claim2} \\
&\leq t^{-2} c(\opt) \label{line:hardround_claim3}
\end{align}
Above \eqref{line:hardround_claim1} is due to \cref{claim:cheapcontribute}, \eqref{line:hardround_claim2} due to \cref{claim:fminLB} and \eqref{line:hardround_claim3} due to \cref{claim:fmaxUB}. Thus the expected cost of cumulative alterations over all time steps $t \in [T]$ is $O(c(\opt))$. To wrap up, \hyperref[thm:buchbinder2009]{Theorem 2.1} implies that ${c(x) = O(\ln n) \cdot c(\opt)}$, hence:
\begin{align*}
\expect{c(R \cup G)} &= \ln \left(T\cdot \frac{\fmax}{\fmin} \right)\cdot c(x) + O(c(\opt)) = O\left( \ln n \ln \left( T \cdot \frac{\fmax}{\fmin} \right) \right) \cdot c(\opt)
\end{align*}
It remains to prove the three claims.
\begin{proof}[Proof of \cref{claim:cheapcontribute}]
	The greedy algorithm repeatedly selects the cheapest element that has nonzero marginal value. If for some element $O_{s}$ it is the case that $c(O_{s}) < c(G_{t})$, then by the time $G_{t}$ is picked, $O_{s}$ must already have been selected, which means $C(t,s) = 0$.
\end{proof}

\begin{proof}[Proof of \cref{claim:fminLB}]
	This is immediate from the definitions:
	\begin{align}
	\sum_{s \in [q]} C(t,s) &= \sum_{s \in [q]} \mutcov(G_{t}; O_{s} \mid G_{1:t-1} \cup O_{1:s-1} \cup R) = \mutcov(G_{t}; N \mid G_{1:t-1} \cup R) \label{line:hardround_chainRule}\\
	&= f_{G_{1:t-1} \cup R}(G_{t}) \geq \fmin \label{line:hardround_fmin}
	\end{align}
	where step \eqref{line:hardround_chainRule} used the \nameref{fact:chainRule}, and the final step \eqref{line:hardround_fmin} used the definition of $\fmin$.
\end{proof}

\begin{proof}[Proof of \cref{claim:fmaxUB}]
	Again by definition:
	\begin{align*}
	\sum_t C(t,s) &= \sum_{t \in [p]} \mutcov(G_{t}; O_{s} \mid G_{1:t-1} \cup O_{1:s-1} \cup R) = \mutcov(N; O_{s} \mid O_{1:s-1} \cup R) \leq \mutcov(N; O_{s} \mid R)
	\end{align*}
	Where the inequality follows from the submodularity of $f$ and $\mutcov(N ; O_s \mid O_{1:s-1} \cup R) = f(O_s \mid O_{1:s-1} \cup R)$. Taking expectation over $R$:
	\begin{align*}
	\expectarg{\sum_t C(t,s)}{R} &\leq \expectarg{\mutcov(O_s; N \mid R)}{R} \leq t^{-2}\left(\frac{\fmin}{\fmax}\right) f(O_{s})
	\leq t^{-2}\fmin
	\end{align*}
	The second inequality
	follows from (\ref{line:onlineSSC2_mutCovRoundingWorks}), and the third 
	from the definition of $\fmax$.
\end{proof}
Given the claims, the proof of \hyperref[thm:slowOnlineSSC2]{Theorem 3.1} is complete.
\end{proof}

It remains to show that the algorithm of this section and the previous
one can be made efficient; this we do in the following section.

\section{Polynomial Runtime}

\label{subsec:polytime}

The algorithms in the two previous sections na{\"\i}vely run in
exponential time, and here are some of the barriers to making them efficient. For one, the algorithm of~\cite{buchbinder2009online} runs in time linear in the number of constraints added, and naively both $(P^{(t)})$ and $(Q^{(t)})$ add an exponential number of constraints at every step. The standard method for solving LPs with many constraints is to construct an efficient separation oracle that, given a point as input, either confirms that it is a feasible solution or outputs a violated constraint. Given such an oracle, we could check if our current solution is feasible, if so output the solution, if not feed the violated constraint to \cite{buchbinder2009online} and repeat.

However, a polynomial-time separation oracle is unlikely to exist since
evaluating $f^*$ is \APX-hard~\cite[Section
3.7.2]{vondrak2007submodularity}. Moreover, even assuming access to a
separation oracle, we would additionally need to argue that we can make
do with a polynomial number of calls to the oracle. However,
the algorithm of \cite{buchbinder2009online} makes no a priori
guarantees about how many constraints it will need to fix before it
produces a feasible solution (in contrast to the ellipsoid algorithm).

Our solution is to avoid solving the fractional LP directly. Since we
merely want to find a feasible integer solution with bounded expected cost, we combine rounding and separation. We show that every LP solution either guarantees random rounding will make progress, or violates an efficiently computable constraint by a large margin. Finally, we show that fixing the large-margin violations requires a large change to the LP solution; by a potential function argument this can only be done a polynomial number of times before the solution is necessarily feasible.

We describe this strategy in the context of the simpler LP $P^{(t)}$. Since $Q^{(t)}$ (when $f$ is $3$-increasing) is the intersection of several of the simpler LP's constraints, one for each submodular function $\mutcov^{(\{v\})}$, this also implies an efficient algorithm for the second setting, because we can search for a large margin violation in \emph{any} of the submodular cover LPs we are intersecting, and round if none is found.

\subsection{Finding Violated Constraints.}

\cref{lem:multilinFstar} implies that for monotone, nonnegative functions $g$, the multilinear extension $G$ (which is precisely the expected coverage of a set obtained by randomized rounding according to $x$) cannot be smaller than $(1- e^{-1}) g^*$. This means that if $G(x)$ is much smaller than $g(N)$, then so is $g^*(x)$, and there must exist $C$ such that $g^{\#C}(x) < g(N)$. Call such a set $C$ \textit{violated}, since we will soon use this as a mechanism for finding violated constraints of \eqref{eq:submodCoverLP}. \cref{lem:multilinFstar} does not immediately imply a constructive algorithm for finding a set $C$ witnessing the violation. We show that the following procedure, inspired by the proof of  \cref{lem:multilinFstar} (Lemma 3.8 of \cite{vondrak2007submodularity}), allows us find such a $C$.

The input is a fractional LP solution $x$. For every element $i\in N$, sample an exponential random variable $h_i$ with parameter $x_i$, and order all the elements $i_1, \ldots i_n$ in increasing order of their corresponding expontential $h_i$. Of all prefixes of this ordering, consider the prefix $Q = i_1, \ldots i_z$ that minimizes $g^{\#Q}(x)$. If $g^{\#Q}(x) < g(N)$, then say that $Q$ is a violated set. Repeat these steps $3/\epsilon \cdot \ln (1 / \delta)$ times or until a violated set is found. If none is found, return ``no constraint found''.
\begin{lemma}
	\label{lem:findViolated}
	There is an algorithm such that given a nonnegative, monotone function $g: 2^N \rightarrow \mathbb{R}^+$  with $g(\emptyset) = 0$, and $x \in [0,1]^n$ with the guarantee that $G(x) < (e(1 + \epsilon))^{-1} g(N)$, with probability $1 - \delta$ the algorithm yields a subset $C \subseteq N$ such that
	\begin{align*}
	g^{\#C}(x) < \left(1 - \frac{\epsilon}{2}\right) \cdot g(N).
	\end{align*}
	Furthermore, the procedure runs in time $\poly\left( \frac{1}{\epsilon} \ln \frac{1}{\delta}\right)$.
\end{lemma}
\begin{proof}
	Consider the set $C(t) := \{i \in N \ : \ h_i < t\}$. 
	Since $\prob{j \in C(1)} = 1 - e^{-x_j} \leq x_j$, we have by monotonicity of $g$ that $G(x) = \expectarg{g(S)}{S \sim x} \geq \expect{g(C(1)}$.
	Now for any $C$,
	\begin{align*}
	&\expect{g_{C(t)}(C(t + dt)) \mid C(t) = C} = \sum_{D \subseteq N} \prob{C(t + dt) = D \mid C(t) = C} \cdot g_C(D) \\
	&= \sum_{D \subseteq N \backslash C} g_C(D) \cdot \prod_{j \in D} (x_j \, dt)  \prod_{j' \not \in D} (1-x_{j'} \, dt) = \sum_{j \in N} g_C(j) x_j dt + O(dt^2) \\
	&= (g^{\#C}(x) - g(C)) dt + O(dt^2).
	\end{align*}
	Dividing by $dt$ and taking an expectation over $C$,
	\begin{align*}
	&\frac{1}{dt} \expect{g(C(t + dt)) - g(C(t))} \geq \expect{g^{\#C(t)}(x)} - \expect{g(C(t))} + O(dt).
	\end{align*}
	Now define $\phi(t) = \expect{g(C(t))}$. Taking the limit of the last expression as $dt \rightarrow 0$ we get that $\frac{d\phi}{dt}= \expect{g^{\#C(t)}(x)} - \phi(t)$. 
	 Using $e^t$ as an integrating factor, for any $t \geq 0$
	\begin{align*}
	e^t \phi(t) = \int_0^t \frac{d(e^z\phi(z))}{dz} dz = \int_0^t e^z \expect{g^{\#C(t)}(x)} dz.
	\end{align*}
	By the mean value theorem, for some value $t^* \in [0, t]$ it holds that
	\begin{align*}e^t \phi(t) = t e^{t^*} \expect{g^{\#C(t^*)}(x)} \geq t \expect{g^{\#C(t^*)}(x)}.\end{align*}
	Now setting $t=1$ and using the fact that $G(x) \geq \phi(1)$, we have $\expect{g^{\#C(t^*)}(x)} \leq e \cdot G(x)$. By construction, $Q = C(t')$ for some $t' = \argmin_t g^{\#C(t)}(x)$. The inequality $g^{\#Q}(x) \leq g^{\#C(t^*)}(x)$ holds regardless of the realizations of $h_1, \ldots, h_n$, hence
	\begin{align*}
	\expect{g^{\#Q}(x)}  \leq \expect{g^{\#C(t^*)}(x)} \leq e \cdot G(x).
	\end{align*}
	By a Markov bound, with probability at least $1 - (1 + \epsilon/2)^{-1} \geq \epsilon / 3$:
	\begin{align*}g^{\#Q}(x) &\leq \left(1+ \frac{\epsilon}{2} \right) e \cdot G(x) < \frac{\left(1+ \frac{\epsilon}{2} \right) e}{(1 + \epsilon) e} \cdot g(N) \leq \left(1 - \frac{\epsilon}{2} \right)\cdot g(N). \end{align*}
	Since we repeat this procedure $k = (3 / \epsilon) \ln 1 / \delta$ times, the probability that after $k$ attempts no violated constraint is found is: 
	$\left(1 - \epsilon / 3\right)^{k} \leq \exp\left( 3 / \epsilon \cdot \epsilon / 3 \ln \delta \right) = \delta$.
\end{proof}

\subsection{Implementing the Algorithm Efficiently.}

\label{sec:efficient}

We modify the proof of \cref{thm:slowOssc} to exploit the approximate separation oracle. This time we do not ever solve \eqref{eq:submodCoverLP} explicitly. Instead, start with a fractional solution $x = 0$. At time $t$, set $f := f^{(t)}_{S_{t-1}}$ (where $S_{t-1
}$ is the solution produced for the previous time step), and perform the following procedure which proceeds in rounds.

In each round, search for violated constraints using the procedure from the last section, with $g= f$ and $\delta = 6\cdot \cmin / (\pi^2  kt^2 \cdot c(N))$. If one is found, feed it to \cite{buchbinder2009online}, update $x$ accordingly and restart the current round. If none is found, perform sample every $i \in N$ with probability $x_i$ (i.e. perform one iteration of randomized rounding).

In total, perform $k := (\ln 1/ \gamma)^{-1} \ln (t^2\cdot f(N) / \fmin)$ rounds of rounding with respect to the appropriate $\Delta \in [0,1]^n$ (where $\gamma := 1 - (e(1 + \epsilon))^{-1}$) and add the result to set $R$. To finish, again greedily add to set $G$ the cheapest element $j \in N$ with non-zero marginal value, until none remain. Return $R \cup G$.

Crucially, we argue that we do not need to fix too many constraints over the course of the algorithm.
\begin{lemma}
	\label{lem:polytime}
	The algorithm above finds a violated constraint at most $O(n / \epsilon)$ times.
\end{lemma}

\begin{proof}
	Consider a round in which the algorithm finds a violated constraint indexed by the set $C$. By \cref{lem:findViolated},
	\begin{align*}
	f_R(C) + \sum_{j\in N} f_{R \cup C}(j) \cdot x_j < \left(1 - \frac{\epsilon}{2}\right) \cdot f_R(N). 
	\end{align*}	
	In order to fix this constraint in this round, \cite{buchbinder2009online} must increase $x$ by a vector $\Delta x$ such that
	\begin{align*}
	\sum_{j \in N} \frac{f_{R\cup C}(j)}{f_R(N)} \cdot (\Delta x)_j > \frac{\epsilon}{2}.
	\end{align*}
	Since $f_{R\cup C}(j) / f_R(N) \leq 1$, it follows that
	\begin{align*}
	\sum_{j \in N} (\Delta x)_j  \geq \sum_{j \in N}  \frac{f_{R\cup C}(j)}{f_R(N)} (\Delta x)_j > \frac{\epsilon}{2}.
	\end{align*}
	Since $x \leq 1$, the total sum $\sum_j x_j$ is bounded by $n$, and thus we can update $x$ at most $2n / \epsilon$ times.
\end{proof}
We can finally conclude with our main theorem.
\mainthm*
\begin{proof}	
	The algorithm always produces a feasible solution and we need to bound its cost. Since we perform $(\ln 1/ \gamma)^{-1} \ln (t^2\cdot f(N) / \fmin)$ rounds of rounding, and $(\ln 1/ \gamma)^{-1} \leq e(1 + \epsilon)$, the expected cost of $R$ is at most $e(1 + \epsilon) \cdot \ln (t^2\cdot f(N) / \fmin) \cdot c(x)$. Next we bound the cost of the greedy phase.  Let $G_t$ denote the elements added to $G$ in iteration $t$.
	
	Let $R_{1:\ell}$ denote the state of $R$ at the end of round $\ell \in [k]$ and let $R_\ell$ denote the set of elements sampled in round $\ell$. By the contrapositive of \cref{lem:findViolated}, for every $v \in N$ and every $\ell \in [k]$, with probability $1 - 6\cdot \cmin / (\pi^2  kt^2 \cdot c(N))$ it holds that:
		\begin{align}
		\expectarg{f_{R_{1:\ell-1}}(R_{\ell})}{R_\ell \sim x} \geq \frac{1}{e(1 + \epsilon)} f_{R_{1:\ell-1}}(N) \label{line:fancyonlineSSC_roundingWorksMutCov}
		\end{align}
		Let $\mathcal{E}$ be the event that \eqref{line:fancyonlineSSC_roundingWorksMutCov} holds for every round of rounding over all time steps $t \in [T]$. By a union bound and the fact that $\sum t^{-2} \leq \pi^2 / 6$, event $\mathcal{E}$ holds with probability at least $1 - \cmin/ c(N)$. Conditioned on this being the case, the triple $x$, $f$ and $R_{1:k}$ together satisfies the conditions of the \nameref{lem:roundingLemma} with $\gamma = 1 - (e(1 - \epsilon))^{-1}$. Hence:
		\begin{align*}
		\expectarg{f_R(N)}{R} &\leq t^{-2}\left(\frac{\fmin}{f(N)}\right) f(N) \\
		&= t^{-2} \fmin 
		\end{align*}
		With this, as in the proof of \cref{thm:slowOssc}, we have  $\expect{c(G_{t}) \mid \mathcal{E}} \leq t^{-2} c(\opt)$, and cumulatively $\expect{\sum_{t\in [T]} c(G_{t}) \mid \mathcal{E}} \leq c(\opt)$. If $\mathcal{E}$ does not hold, we incur a cost of at most $c(N)$. Hence the total expected cost of cumulative alterations is:
		\begin{align*}
		\expect{c(G)} &= \expect{c(G) \mid \mathcal{E}} \cdot \prob{\mathcal{E}} + \expect{c(G) \mid \overline{\mathcal{E}}} \cdot \prob{\overline{\mathcal{E}}} \leq \frac{\cmin}{ c(N)} \cdot c(N) + O(c(\opt)) = O(c(\opt))
		\end{align*}
		Finally, \hyperref[thm:buchbinder2009]{Theorem 2.1} implies that ${c(x) = O(\ln n) \cdot c(\opt)}$. Setting $\epsilon = 1$, we conclude:
		\begin{align*}
		&\expect{c(R \cup G)} = e(1 + \epsilon) \ln \left(T^2\cdot \frac{f(N)}{\fmin} \right) c(x) + O(c(\opt)) \leq O\left( \ln \left(T\cdot \frac{f(N)}{\fmin} \right)\right) c(\opt)
		\end{align*}
		This completes the proof of correctness. The polynomial runtime follows from \cref{lem:polytime}.
\end{proof}

 \section{Note on Changes since Publication} 

 \label{sec:version_notes}

 The original version of this paper that appeared at SODA 2020 claimed \cref{thm:slowOnlineSSC2} without assuming $f \in \mathcal{D}_3^+$. We would like to thank Jonny Gal for pointing out to us an error with \cref{lem:fancyrelaxation}, which is remedied by assuming $f \in \mathcal{D}_3^+$. The assumption further simplifies the proofs of \cref{subsec:betterround}, since the strengthened LP \eqref{eq:newRelaxation} is the intersection of several LPs \eqref{eq:submodCoverLP}. The remaining proofs remain essentially unchanged, though we have also simplified the exposition in \cref{sec:efficient}. Showing \cref{thm:slowOnlineSSC2} for the class of all nonnegative monotone submodular functions remains an interesting open problem.

{\footnotesize
\bibliography{refs}
\bibliographystyle{alpha}
}
		
\appendix

\section{Necessity of Time-Monotonicity}
\label{section:timeMonotonicity}
We note briefly that if we drop the time-monotonicity condition, then no algorithm can achieve a bounded competitive ratio. Consider the sequence $f^{(0)}, f^{(1)}, f^{(2)}$ such that $f^{(0)} \equiv f^{(2)} \equiv 0$, and $f^{(1)}(S) = {\pmb 1}(j \in S)$ for an abitrary choice of $j \in N$. In this case the optimal solution at time $2$ is the empty set, whereas any online algorithm must output a set containing $j$.

\section{Mutual Coverage is not Submodular}

\label{section:appendixStrongSubmod}

We give one example which shows that mutual coverage is in general \textbf{not} submodular as a function of one of its arguments. Furthermore, the example demonstrates that the condition $F_S(x) \geq \alpha f_S(N)$ for all subsets $S \subseteq N$ does not guarantee that $\expectarg{\mutcov^{(y \mid S)}(R)}{R \sim x} \geq \alpha \cdot \mutcov^{(y \mid S)}(N)$. This justifies our use of the stronger LP \eqref{eq:newRelaxation} instead of \eqref{eq:submodCoverLP} in section \cref{subsec:betterround}.

Consider the following function on $\{a, b, y\}$:
\begin{align*}
	\begin{array}{lll}
		f(\emptyset) = 0 & f(b) = 1 & f(b+y) = 2 \\
		f(a + b+ y) = 10 & f(y) = 1 & f(a + b) = 10  \\
		& f(a) = 9 & f(a + y) = 10 \\
	\end{array}
\end{align*}

The function $f$ is submodular, but $\mutcov^{(y)}(a) -  \mutcov^{(y)}(\emptyset) < \mutcov^{(y)}(b + a) -  \mutcov^{(y)}(b)$, since:
\begin{align*}
	&\mutcov^{(y)}(a) -  \mutcov^{(y)}(\emptyset) \\
	&= f( a) - f(y + a) - f(\emptyset) + f(y) \\
	&= 0 \\
	&\mutcov^{(y)}(b + a) -  \mutcov^{(y)}(b) \\
	&= f(b + a) - f(a + b + y) - f(b) + f(b + y) \\
	&= 1
\end{align*}
Now consider the vector $x = [x_a, x_b, x_y] = [3/4, 1/2, 0]$. One can check that $F_S(x) \geq 1/2 \cdot f_S(N)$ for all $S \subseteq N$ but:
\begin{align*}
	&\expectarg{\mutcov^{(y)}(R)}{R \sim x} \\
	&=  \frac{3}{8} \mutcov^{(y)}(a + b) + \frac{3}{8} \mutcov^{(y)}(a) + \frac{1}{8} \mutcov^{(y)}(b) + \frac{3}{8} \mutcov^{(y)}(\emptyset) \\
	&=\frac{3}{8} < \frac{1}{2} \mutcov^{(y)}(N)
\end{align*}

\end{document}